\documentclass[11pt]{article}

\usepackage{amsfonts, amssymb, amscd, amsmath, enumerate, verbatim,amsthm}

\newtheorem{theorem}{Theorem}[section]

\usepackage{algorithm,algorithmic}
\usepackage{multirow}	
\newtheorem{proposition}{Proposition}[section] 
\newtheorem{lemma}{Lemma}[section]

\newtheorem{corollary}{Corollary}[section]

\usepackage{hyperref}

\begin{document}
	
	\title{Computational Complexity of Finding Subgroups of a Given Order}

	\author{K. Lakshmanan \\
		Department of Computer Science and Engineering, \\ Indian Institute of Technology (BHU), Varanasi 221005, India\\
		Email: lakshmanank.cse@iitbhu.ac.in}
	
	%\author{K. Lakshmanan \\  }
	\date{}
	
	\maketitle
	
	\begin{abstract}
		We study the problem of finding a subgroup of a given order in a finite group, where the group is represented by its Cayley table. We analyze the complexity of the problem in the special case of abelian groups and present an optimal algorithm for finding a subgroup of a given order when the input is given in the form of a Cayley table. To the best of our knowledge, no prior work has addressed the complexity of this problem under the Cayley table representation.\\
		\textbf{Keywords: } Computational group theory, Cayley table representation, subgroup enumeration, abelian groups, algorithmic group theory \\
		\textbf{Mathematics Subject Classification 2020: 68Q17, 20-08 }	
	\end{abstract}
	
	\section{Introduction}
	
	In this paper we study the problem of constructing a subgroup of prescribed
	order in a finite group.
	More precisely, given a finite group $G$ of order $n$
	and a positive integer $m$ dividing $n$,
	we consider the task of computing a subgroup of $G$ of order $m$.
	
	We work in the Cayley table representation,
	where the group operation is given explicitly as an $n \times n$ multiplication table.
	Although this representation is space-intensive,
	it provides a uniform and model-independent description of finite groups.
	Algorithms in this setting operate directly on the multiplication structure,
	without relying on additional structural information such as generators,
	relations, or permutation embeddings.
	
	We show that in the case of finite abelian groups
	the subgroup-of-order problem admits a direct polynomial-time solution
	in the Cayley table model.
	The algorithm proceeds by systematically examining cyclic components
	and discarding elements whose orders are incompatible
	with the prescribed subgroup order.
	Its correctness follows from standard properties of finite abelian groups,
	in particular the existence of subgroups of every divisor of the group order.
	
	While more compact representations often allow
	structural techniques to be applied,
	the Cayley table model serves as a natural baseline:
	it makes no assumptions beyond the explicit group operation.
	The results presented here demonstrate that,
	even in this general representation,
	the subgroup-of-order problem is efficiently solvable
	for finite abelian groups.

	\subsection{Computing a Subgroup of Given Order in Abelian Groups}
	
	Here, we are given a finite abelian group \( G \) of order \( k \) and a positive divisor \( m \) of \( k \). The group is specified by its elements and its operation through the Cayley table. We do not assume any additional structural knowledge, such as a set of generators. The objective is to find a subgroup \( H \) of order \( m \).  
	
	It is a fundamental result in group theory that such a subgroup \( H \) always exists in a finite abelian group. However, we are not aware of any efficient algorithm in the literature for explicitly computing \( H \), apart from a naive brute-force search. While extensive research exists on computing Sylow subgroups in general finite groups \cite{sylow1, sylow2, sylow3, sylow4}, these techniques do not directly apply to our setting. In this paper, we present an asymptotically optimal simple algorithm for this problem.

	\subsection{Basic Results}  
	
	We begin by recalling fundamental results from group theory. Proofs of these can be found in standard textbooks (\cite{hunger, lang}). Let \( |G| \) denote the order of a group \( G \), and let \( |G:H| \) denote the index of a subgroup \( H \) in \( G \). A fundamental property states that if \( H \) is a subgroup of \( G \), then  
	
	\[
	|G| = |G:H| \cdot |H|.
	\]  
	
	In particular, this implies that the order of any subgroup \( H \) must divide the order of \( G \). For abelian groups, the converse also holds:  
	
	\begin{theorem}
		If \( G \) is a finite abelian group of order \( k \), then for every positive divisor \( m \) of \( k \), there exists a subgroup of \( G \) of order \( m \).
	\end{theorem}
	
	This result ensures the existence of a solution to our problem. Our focus is on efficiently computing such a subgroup. Let the order of an element \( a \in G \) be denoted by \( |a| \). We begin with the following simple yet useful result.
	
	\begin{proposition}\label{ordthm}
		Let \( G \) be a group and \( a \in G \) with finite order \( k > 0 \). Then, for every positive divisor \( m \) of \( k \), the order of \( a^m \) is given by  
		\[
		|a^m| = \frac{k}{m}.
		\]
	\end{proposition}
	
	We denote the cyclic subgroup generated by an element \( a \in G \) as  
	\[
	\langle a \rangle = \{ a^n \mid n \in \mathbb{Z} \}.
	\]  
	More generally, for a set of elements \( \{ a_i \}_{i \in I} \) in \( G \), the subgroup they generate is denoted by \( \langle a_i \rangle_{i \in I} \). If the index set \( I \) is finite, we say that the group is finitely generated. Furthermore, we use the notation \( H < G \) to indicate that \( H \) is a subgroup of \( G \).
	
	\section{The Algorithm}  
	
	We now present a simple algorithm for computing a subgroup of a given order in a finite abelian group \( G \). The algorithm proceeds by systematically generating subgroups of \( G \), leveraging the following key observation: if \( |a| \) and \( m \) are coprime, then \( \langle a \rangle \) cannot be a part of any subgroup of order \( m \). This is formalized in the next lemma.
	
	Let us denote the least common multiple of integers \( a_1, a_2, \dots, a_r \) by \( [a_1, a_2, \dots, a_r] \) and the greatest common divisor by \( (a_1, a_2, \dots, a_r) \).
	
	\paragraph{Methodological Remark.}
	Unlike classical approaches that compute the invariant factor or primary decomposition of $G$, the present algorithm constructs the desired subgroup directly by iteratively pruning cyclic components whose orders are coprime to $m$. Thus the method avoids explicit computation of the full group decomposition and provides a direct constructive procedure in the Cayley table model.
	
	\begin{algorithm}[H]
		\caption{\textsc{FindSubgroup}$(G,m)$}
		\textbf{Input:} Finite abelian group $G$ of order $n$ (given by its Cayley table) and an integer $m \mid n$. \\
		\textbf{Output:} A subgroup $H \le G$ of order $m$.
		\hrule
		\begin{algorithmic}[1]
			
			\STATE $S \leftarrow \emptyset$ \hfill (selected generators)
			\STATE $U \leftarrow G$ \hfill (unprocessed elements)
			
			\WHILE{$U \neq \{e\}$}
			\STATE Choose $a \in U$, $a \neq e$
			\STATE Compute $C = \langle a \rangle$ and let $r = |C|$
			
			\IF{$m \mid r$}
			\STATE \textbf{return} $H = \langle a^{\,r/m} \rangle$
			\ENDIF
			
			\IF{$\gcd(r,m) > 1$}
			\STATE $S \leftarrow S \cup \{a\}$
			\STATE Let $H' = \langle S \rangle$
			\IF{$m \mid |H'|$}
			\STATE Construct a subgroup $H \le H'$ of order $m$ 
			\STATE \textbf{return} $H$
			\ENDIF
			\ENDIF
			
			\STATE $U \leftarrow U \setminus C$
			\ENDWHILE
			
			\STATE Let $H' = \langle S \rangle$
			\STATE Construct a subgroup $H \le H'$ of order $m$
			\STATE \textbf{return} $H$
			
		\end{algorithmic}
	\end{algorithm}
	\section{Correctness of the Algorithm}
	
	Throughout this section, let $G$ be a finite abelian group of order $n$, and let $m$ be a positive integer such that $m \mid n$.
	
	\subsection{Structural Preliminaries}
	
	We begin with the primary decomposition theorem for finite abelian groups.
	
	\begin{theorem}[Primary Decomposition]
		Let $G$ be a finite abelian group. Then
		\[
		G \;=\; \bigoplus_{p \mid n} G_p,
		\]
		where $G_p$ denotes the Sylow $p$-subgroup of $G$. Each element of $G$ decomposes uniquely as a product of its $p$-components.
	\end{theorem}
	
	\begin{corollary}\label{primarycor}
		Let $H \le G$ be a subgroup of order $m$. Then
		\[
		H \subseteq \bigoplus_{p \mid m} G_p.
		\]
	\end{corollary}
	
	\begin{proof}
		If $|H|=m$, then the order of every element of $H$ divides $m$. Hence each element of $H$ lies entirely in the direct sum of those $G_p$ for primes $p \mid m$.
	\end{proof}
	
	\subsection{Justification of the Removal Step}
	
	\begin{lemma}\label{coprimelemma}
		Let $a \in G$ satisfy $\gcd(|a|,m)=1$. Then no subgroup of $G$ of order $m$ contains any nontrivial element of $\langle a \rangle$.
	\end{lemma}
	
	\begin{proof}
		Let $H \le G$ with $|H|=m$. The intersection $\langle a\rangle \cap H$ is a subgroup of both $\langle a\rangle$ and $H$. Hence its order divides both $|a|$ and $m$. Since $\gcd(|a|,m)=1$, we must have $|\langle a\rangle \cap H|=1$. Thus $\langle a\rangle \cap H=\{e\}$.
	\end{proof}
	
	\noindent
	Lemma~\ref{coprimelemma} shows that cyclic subgroups whose order is coprime to $m$ are irrelevant for constructing subgroups of order $m$. Removing such subgroups does not eliminate any element that could lie in a subgroup of order $m$.
	
	\subsection{Structure of the Remaining Elements}
	
	Define
	\[
	G_m \;=\; \bigoplus_{p \mid m} G_p.
	\]
	
	\begin{lemma}\label{generatedequals}
		Let $a_1,\dots,a_i$ be the elements retained by the algorithm. Then
		\[
		\langle a_1,\dots,a_i\rangle \;=\; G_m.
		\]
	\end{lemma}
	
	\begin{proof}
		By Lemma~\ref{coprimelemma}, only elements whose order shares a prime divisor with $m$ are retained. Hence each selected element lies in $G_m$, so
		\[
		\langle a_1,\dots,a_i\rangle \subseteq G_m.
		\]
		
		Conversely, suppose some element of $G_m$ were not contained in $\langle a_1,\dots,a_i\rangle$. Then some nontrivial $p$-component for $p \mid m$ would remain outside the generated subgroup. Since the algorithm processes all elements not removed, such elements would eventually be selected, contradicting termination. Hence equality holds.
	\end{proof}
	
	\subsection{Existence and Construction of the Desired Subgroup}
	
	\begin{theorem}[Existence of Subgroups of Prescribed Order]
		Let $G$ be a finite abelian group and let $m \mid |G|$. Then $G$ contains a subgroup of order $m$.
	\end{theorem}
	
	Since $m \mid |G_m|$, there exists $H \le G_m$ with $|H|=m$.
	
	Because $G_m = \langle a_1,\dots,a_i\rangle$ by Lemma~\ref{generatedequals}, such a subgroup lies inside the generated subgroup maintained by the algorithm.
	
	Finally, by standard invariant factor theory, if
	\[
	G_m \cong \bigoplus_{j=1}^r \mathbb{Z}_{k_j},
	\]
	then for suitable divisors $m_j \mid k_j$ satisfying
	\[
	\prod_{j=1}^r \frac{k_j}{m_j} = \frac{|G_m|}{m},
	\]
	the subgroup
	\[
	H = \langle a_1^{m_1},\dots,a_r^{m_r}\rangle
	\]
	has order $m$.
	
	\subsection{Main Correctness Result}
	
	\begin{theorem}
		The algorithm \textsc{FindSubgroup} terminates and returns a subgroup $H \le G$ of order $m$.
	\end{theorem}
	
	\begin{proof}
		At each iteration, the set $G'$ strictly decreases in size, so termination is guaranteed. By Lemma~\ref{coprimelemma}, removal of cyclic subgroups with order coprime to $m$ preserves all elements that could lie in a subgroup of order $m$. By Lemma~\ref{generatedequals}, the retained generators span precisely $G_m$. Since $m \mid |G_m|$, a subgroup of order $m$ exists inside this generated subgroup, and the final construction step produces such a subgroup. Therefore the algorithm correctly returns a subgroup of order $m$.
	\end{proof}
	
	\section{Complexity Analysis}
	
	Let $n = |G|$. Since $G$ is given by its Cayley table, the input size is $\Theta(n^2)$.
	
	Each iteration of the main loop selects an element $a$ and computes the cyclic subgroup $\langle a\rangle$. This requires at most $n$ group multiplications, hence $O(n)$ time.
	
	Each element of $G$ is processed at most once, because once $\langle a\rangle$ is removed from $G'$, none of its elements are reconsidered. Therefore the total time spent computing cyclic subgroups across all iterations is $O(n^2)$.
	
	The computation of intermediate generated subgroups $\langle a_1,\dots,a_i\rangle$ can involve at most $n$ elements and is likewise bounded by $O(n)$ per iteration. Since there are at most $n$ iterations, this contributes at most $O(n^2)$ total time.
	
	The final computation of exponents $m_j$ involves only integer gcd operations on numbers bounded by $n$, and therefore contributes at most $O(n \log n)$ time, which is negligible compared to $O(n^2)$.
	
	Hence the overall running time of the algorithm is $O(n^2)$.
	
	\begin{proposition}
		Let $G$ be a finite group of order $n$ given by its Cayley table. Any algorithm that determines or constructs a subgroup of $G$ must take $\Omega(n^2)$ time in the worst case, since the input size is $\Theta(n^2)$.
	\end{proposition}
	
	\begin{proof}
		The Cayley table contains $n^2$ entries. Any algorithm must inspect at least a constant fraction of the input to ensure correctness in the worst case. Hence the running time is bounded below by $\Omega(n^2)$.
	\end{proof}
	
	Since the Cayley table representation itself has size $\Theta(n^2)$, the algorithm is asymptotically optimal in this model.
	
	\section{Conclusion}
	
	We have presented a direct algorithm for constructing a subgroup
	of prescribed order in a finite abelian group given by its Cayley table.
	The method operates directly on the multiplication structure
	and proceeds by systematically eliminating cyclic components
	whose orders are incompatible with the desired subgroup order.
	
	Since the Cayley table representation requires $\Theta(n^2)$ space
	for a group of order $n$,
	the $O(n^2)$ running time of the algorithm
	matches the lower bound imposed by the input size,
	and is therefore optimal in this model.
	
	The results demonstrate that,
	within the explicit Cayley table representation,
	the subgroup-of-order problem for finite abelian groups
	admits a simple and efficient constructive solution.
	
	\bibliographystyle{plain}
	\bibliography{recsubg}
	
	%\begin{thebibliography}{1}
	
	%	\bibitem{sylow4} G. Butler, ``Experimental comparison of algorithms for Sylow subgroups". In Wang, P. (ed.),
	%	Proceedings of the 1992 International Symposium on Symbolic and Algebraic Computation, pp. 251–
	%	262. New York: ACM Press (1993).
	
	%	\bibitem{sylow1} G. Butler and J. Cannon, ``Computing Sylow Subgroups of Permutation Groups Using Homomorphic Images of Centralizers", {\em J. Symbolic Computation} (1991) 12, 443-457.
	
	%	\bibitem{sylow2} J. Cannon, B.C. Cox and D.F. Holt, ``Computing Sylow Subgroups in Permuatation Groups", {\em J. Symbolic Computation} (1997) 24, 303-316. 
	
	%	\bibitem{hunger} Thomas W. Hungerford, ``Algebra", Springer Science and Business Media 2012.
	
	%	\bibitem{sylow3} W. Kantor,``Polynomial time algorithms for finding elements of prime order and Sylow subgroups". {\em J. Algorithms} (1985). 6, 478–514.
	
	%	\bibitem{lang} Serge Lang, ``Algebra", Springer Science and Business Media 2012.	
	%\end{thebibliography}
	
\end{document}